\documentclass[conference]{IEEEtran}
\ifCLASSINFOpdf
\else
\fi

\usepackage{epstopdf}
\usepackage{amsmath,mathrsfs,amsfonts}
\usepackage{algorithm}
\usepackage{algorithmic}
\usepackage{graphicx}
\usepackage{subfigure}
\newtheorem{theorem}{Theorem}
\newtheorem{lemma}{Lemma}

\begin{document}

\title{A New Space for Comparing Graphs}

\author{\IEEEauthorblockN{\textbf{Anshumali Shrivastava} }
\IEEEauthorblockA{Department of Computer Science\\Computing and Information Science\\
Cornell University\\
Ithaca, NY 14853, USA\\
Email: anshu@cs.cornell.edu}
\and
\IEEEauthorblockN{\textbf{Ping Li}}
\IEEEauthorblockA{Department of Statistics and Biostatistics\\
Department of Computer Science\\
Rutgers University\\
Piscataway, NJ 08854, USA\\
Email: pingli@stat.rutgers.edu}}


%


\maketitle

\begin{abstract}
Finding a new mathematical representations for graph, which allows direct comparison between different graph structures, is  an open-ended research direction. Having such a representation is the first prerequisite for a variety of machine learning algorithms like classification, clustering, etc., over graph datasets. In this paper, we propose a symmetric positive semidefinite matrix with the $(i,j)$-{th} entry equal to the covariance between normalized vectors $A^ie$ and $A^je$ ($e$ being vector of all ones) as a representation for graph with adjacency matrix $A$. We show that the proposed matrix representation encodes the spectrum of the underlying adjacency matrix and it also contains information about the counts of small sub-structures present in the graph such as triangles and small paths. In addition, we show that this matrix is a \emph{``graph invariant"}. All these properties make the proposed matrix a suitable object for representing graphs.

\noindent The representation, being a covariance matrix in a fixed dimensional metric space, gives a mathematical embedding for graphs. This naturally leads to a measure of similarity on graph objects. We define similarity between two given graphs as a  Bhattacharya similarity measure between their corresponding covariance matrix representations. As shown in our experimental study on the task of social network  classification, such a similarity measure outperforms other widely used state-of-the-art methodologies. Our proposed method is also computationally efficient. The computation of both the matrix representation and the similarity value can be performed in operations linear in the number of edges. This makes our  method scalable in practice.

\noindent We believe our  theoretical and empirical results  provide  evidence for studying truncated power iterations, of the adjacency matrix, to characterize social networks.

\end{abstract}


%
\IEEEpeerreviewmaketitle

\section{Introduction}

The study of social networks is becoming increasingly popular.  A whole new set of information about an individual is gained by analyzing the data that is derived from his/her social network. Personal social network of an individual consisting only of neighbors and connections between them, also known as ``ego network", has recently grabbed significant attention~\cite{mcauley2012,ugander2013}. This new view of the gigantic incomprehensible social network as a collection of small informative overlapping ego networks generates a huge collection of graphs, which leads to a closer and more tractable investigation.

This enormous collection of ego networks, one centered at each user, opens doors for many interesting possibilities which were not explored before.  For instance, consider the scientific collaboration ego network of an individual. It is known that collaboration follows different patterns across different fields~\cite{newman2001}. Some scientific communities are more tightly linked among themselves compared to other fields having fewer dependencies among the collaborators. For instance, people working in experimental high energy physics are very much dependent on specialized labs worldwide (for example CERN), and hence it is more likely that scientists in this field have a lot of collaboration among themselves. Collaboration network in such a scientific domain will exhibit more densely connected network compared to other fields where people prefer to work more independently.

The peculiarity in the collaboration network gets reflected in the ego network as well. For an individual belonging to a more tightly connected field, such as high energy physics, it is more likely that there is collaboration among the individual's coauthors.  Thus, we can expect the collaboration ego network of an individual to contain information about the characteristic of his/her research. By utilizing this information, it should be possible to discriminate (classify) between scientists  based on the ego networks of their collaboration. This information can be useful in many applications, for instance, in user based recommendations~\cite{MoralesGL12,he2010social}, recommending jobs~\cite{Paparrizos:2011}, discovering new collaborations~\cite{chen2011collabseer}, citation recommendations~\cite{he2010context}.

The focus of this paper is on social network classification or equivalently graph classification. The first prerequisite for classifying networks is having the ``right" measure of similarity between different graph structures. Finding such a similarity measure is directly related to the problem of computing meaningful mathematical embedding  of network structures. In this work, we address this fundamental problem of finding an appropriate tractable mathematical representation for graphs.

There are many  theories that show the peculiarities of social networks~\cite{strogatz2001exploring,broder2000graph,kumar2010structure}. For instance, it is known that the spectrum of the  adjacency matrix of a real-world graph is very specific. In particular, it has been observed that scale-free graphs develop a triangle like spectral density with a power-law tail, while small-world graphs have a complex spectral density consisting of several sharp peaks~\cite{farkas2001spectra}. Despite such insight into social graph structures, finding a meaningful mathematical representation for these networks where  various graph structures can be directly compared or analyzed in a  common space  is  an understudied area.  Note that the eigenvalues of a graph, which characterize its spectrum, are not directly comparable. Moreover, the eigenvalues as feature vector is not a common space because a larger graph will have more number of significant eigenvalues compared to a smaller graph.

Recently it was shown that representing graphs as a normalized frequency vector, by counting the number of occurrences of various small $k$-size subgraphs ($k =$ 3 or 4), leads to an informative representation~\cite{shervashidze2009efficient, ugander2013}. It was shown that this representation naturally models known distinctive social network characteristics like the ``\emph{triadic closure}". Computing similarity between two graphs as the inner product between such frequency vector representations leads to the state-of-the-art social network classification algorithms.

It is not clear that a histogram based only on counting small subgraphs sufficiently captures all the properties of a graph structure. Only counting small $k$-subgraphs ($k =$ 3 or 4) loses information. It is also not very clear what is the right size  $k$ that provides the right tradeoff between computation and expressiveness. For instance, we observe that (see Section~\ref{sec:comp}) $k =5$ leads to  improvement over $k = 4$ but it comes with a significant computational cost. Although, it is known that histograms based on counting subgraphs of size $k$ can be reasonably approximated by sampling few induced subgraphs of size $k$,  counting subgraphs with $k \ge 5$ is still computationally expensive because it requires testing the given sampled subgraph with the representative set of graphs for isomorphism (see Section~\ref{sec:comp}).  Finding other rich representation for graph, which aptly captures its behavior and is also computationally inexpensive, is an important research direction.

One challenge in meaningfully representing graphs in a common space is the basic requirement that isomorphic graphs should map to the same object.  Features based on counting substructures, for example the frequency of subgraphs,  satisfy this requirement by default but ensuring this property is not trivial if we take a non-counting based approach.

\noindent{\bf Our Contributions:}  We take an alternate route and characterize graph based on the truncated power iteration of the corresponding adjacency matrix $A$, starting with the vector of all ones denoted by $e$. Such a power iteration generates vector $A^ie$ in the $i^{th}$ iteration. We argue that the covariance between vectors of the form $A^ie$ and $A^je$, given some $i$ and $j$, is an informative feature for a given graph. We show that
these covariances are ``graph invariants". They also contain information about the spectrum of the adjacency matrix which is an important characteristic of a random graph~\cite{chung2003spectra}. In addition, taking an altogether different view, it can be shown that these covariances are also related to the counts of small local structures in the given graph.

Instead of a histogram based feature vector representation, we represent graph as a symmetric positive semidefinite covariance matrix $C^A$ whose $(i,j)$-th entry is the covariance between vectors $A^ie$ and $A^je$. To the best of our knowledge this is the first representation of its kind. We further compute  similarity between two given graphs as the standard Bhattacharya similarity between the corresponding covariance matrix representations. Our proposal follows a simple procedure involving only matrix vector multiplications and summations. The entire procedure can be computed in time linear in the number of edges which makes our approach scalable in practice. Similarity based on this new representation outperforms exiting methods on the task of real social network classification.  For example, using the similarity based on the histogram based representation, by counting the number of small subgraphs, performs poorly compared to the proposed measure. These encouraging results  provide  motivation for studying power iteration of the adjacency matrix for social network analysis.

In addition to the above contributions, this paper  provides some  interesting insights in the domain of the collaboration networks. We show that it is possible to distinguish researchers working in different experimental physics sub-domains  just based on the ego network of the researcher's scientific collaboration. To the best of our knowledge this is the first work that explores the information contained in the ego network of scientific collaborations. The results presented could be of independent interest in itself.

\section{Notations and Related Concepts}\label{sec_notation}

The focus of this work is on undirected, unweighted and connected graphs. Any graph $G = \{V, E\}$, with $|V| =n$  and $|E| = m$, is represented by an adjacency matrix $A \in \mathbb{R}^{n \times n}$, where $A_{i,j} = 1$ if and only if $(i,j) \in E$. For a matrix A, we use $A_{(i),(:)} \in \mathbb{R}^{1 \times n}$ to denote the $i^{th}$ row of matrix $A$, while $A_{(:),(j)} \in \mathbb{R}^{n \times 1}$ denotes its $j^{th}$ column.  We use $e$ to denote the vector with all components being 1. Dimension of vector $e$ will be implicit depending on the operation. Vectors are by default column vectors ($\mathbb{R}^{n \times 1}$). The transpose of a matrix $A$ is denoted by $A^T$, defined as $A^T_{i,j} = A_{j,i}$. For a vector $v$, we use $v(i)$ to denotes its $i^{th}$ component.

Two graphs $G$ and $H$ are {\bf \emph{isomorphic}} if there is a bijection between the vertex sets of $G$ and $H$,  $f \colon V(G) \to V(H)$,  such that any two vertices $u, v \in V^G$  are adjacent in $G$ if and only if $f(u)$ and $f(v)$ are adjacent in $H$. Every permutation $\pi:\{1,2,..,n\} \rightarrow \{1,2,..,n\}$ is associated with a corresponding {\bf permutation matrix $P$}. The matrix operator $P$ left multiplied to  matrix $A$ shuffles the rows according to $\pi$ while right multiplication with $P$ shuffles the columns, i.e., matrix $PA$ can be  obtained by shuffling the rows of $A$ under $\pi$ and $AP$ can be obtained by shuffling the columns of $A$ under $\pi$.  Given an adjacency matrix $A$, graphs corresponding to adjacency matrix $A$ and $PAP^T$ are isomorphic, i.e., they represent the same graph structure.  A property of graph, which does not change under the transformation of reordering of vertices is called {\bf \emph{Graph Invariant}}.

For adjacency matrix $A$, let $\lambda_1 \ge \lambda_2 \ge ... \ge \lambda_n$ be the eigenvalues  and $v_1, v_2, ... v_n $ be the corresponding eigenvectors. We denote the component-wise sum of the eigenvectors by $s_1, s_2, ..., s_n$, i.e., $s_i$ denotes the component-wise sum of  $v_i$. A path $p$ of length $L$ is  a sequence of $L+1$ vertices $\{v_1,  v_2, ... v_{L+1}\}$, such that there exists an edge between any two consecutive terms in the sequence, i.e.,  $(v_i,v_{i+1}) \in E$  $\forall  i \in \{1,2,...,L\}$. An edge $x$ belongs to a path $p =\{v_1,  v_2, ... v_{L+1}\}$ if there exists $i$ such that $x = (v_i,v_{i+1})$.

In our analysis, we can have paths with repeated nodes, i.e. we will encounter paths where $v_i = v_j$ for $i \ne j$. A path will be called \emph{``simple"} if there is no such repetition of nodes. Formally, a {\bf \emph{simple path}} of length $L$ is a path of length $L$, such that,   $v_i \ne v_j$ whenever $i \ne j$. Two paths $p$ and $q$ are different if there exist an edge $e$, such that  either of the two conditions $(e \in p \text{ and } e \notin q)$ or $(e \in q \text{ and } e \notin p)$ holds, i.e., there exists one edge which is not contained in one of the paths but contained in the other. We denote the number of all the different \emph{``simple paths"} of length $L$ in a given graph by ${\bf P_L}$ and the total number of triangles by ${\bf \Delta}$. For clarity we will use [] to highlight scalar quantities such as $[e^tAe]$.

\section{Graphs as a  Positive Semidefinite Matrix}

A graph is fully described by its adjacency matrix. A good characterization of a matrix operator is a small history of its \emph{power iteration}. \emph{Power iteration} of a matrix $A \in \mathbb{R}^{n \times n}$ on a given starting vector $v \in \mathbb{R}^{n \times 1}$ computes normalized $A^iv \in \mathbb{R}^{n \times 1}$ in the $i^{th}$ iteration.

In one of the early  results~\cite{krylov1931}, it was shown that the characteristic polynomial of a matrix can be computed by using the set of vectors generated from its truncated power iterations, i.e.,  $\{v, Av, A^2v, ..., A^kv\}$. This set of vectors are more commonly known as the ``$k$-order Krylov subspace'' of matrix $A$.
The {``Krylov subspace"} leads to some of the fast linear algebraic algorithms for sparse matrices. In web domain, power iteration are used in known  algorithms including  {\em Page-rank} and {\em HITS}~\cite{kleinberg1999authoritative}. It is also known~\cite{lin2010power} that a truncated power iteration of the data similarity matrix leads to  informative feature representation for clustering. Thus, the $k$-order Krylov subspace for some appropriately chosen $k$ contains sufficient information to  describe the associated matrix.

To  represent graphs in a common mathematical space, it is a basic requirement that two isomorphic graphs should map to the same object. Although the $k$-order Krylov subspace characterizes the adjacency matrix, it can not be directly used as a common representation for the associated graph, because it  is sensitive to the reordering of nodes. Given a permutation matrix $P$,  the $k$-order Krylov subspaces of $A$ and $PAP^T$  can be very different. In other words the mapping $M : A \rightarrow \{v, Av, A^2v, ..., A^kv\}$ is not a \emph{``graph invariant''} mapping. Note that $A$ and $PAP^T$ represent same graph structure with different ordering of nodes and hence are same entities from graph perspective but not from the matrix perspective.

It turns out that if we use $v = e$, the vector of all ones, then the covariances between the different vectors in the power iteration are \emph{``graph invariant"} (see Theorem~\ref{the:inv}), i.e., their values do not change with the spurious reordering of the nodes.  We  start by defining our covariance matrix representation for the given graph, and the algorithm to compute it. In later sections we will argue why such a representation is suitable for discriminating between  graph structures.

Given a graph with adjacency matrix $A \in \mathbb{R}^{n \times n}$ and a fixed number $k$,  we compute the first $k$ terms of power iteration, which generates normalized vectors of the form $A^ie$ $i \in\{1,2, ..., k\}$. Since we start with $e$, we choose to normalize the sum equal to $n$ for the ease of analysis.   After generating $k$ vectors, we compute matrix $C^A \in \mathbb{R}^{k\times k}$ where $C^A_{i,j} = Cov(\frac{nA^ie}{||A^ie||_1},\frac{nA^je}{||A^je||_1})$, as summarized in Algorithm~\ref{alg:cov}.

\begin{algorithm}[h!]
\caption{\emph{CovarianceRepresentation(A,k)}}
\label{alg:cov}
\begin{algorithmic}
\STATE {\bfseries Input:} Adjacency matrix $A\in\mathbb{R}^{n \times n}$, $k$, the number of power iterations.\ Initialize  $x^0 = e \in \mathbb{R}^{n \times 1}.$
\FOR{$t=1$ {\bfseries to} $k$}
\STATE \hspace{0.3in}$M_{(:),(t)} =   n \times \frac{Ax^{t-1}}{||Ax^{t-1}||_1}$,  \hspace{0.3in} $x^t  = M_{(:),(t)}$
\ENDFOR
\STATE $\mu =e \in \mathbb{R}^{k \times 1}$

\STATE $C^A = \frac{1}{n} \sum_{i=1}^n(M_{(i),(:)} - \mu)(M_{(i),(:)} - \mu)^T$

\RETURN  $C^A \in \mathbb{R}^{k \times k}$
\end{algorithmic}
\end{algorithm}%

Algorithm~\ref{alg:cov} maps a given graph to a positive semidefinite matrix, which  is a graph invariant.

\begin{theorem}
\label{the:inv}
$C^A$ is symmetric positive semidefinite. For any given permutation matrix $P$ we have $C^A = C^{PAP^T}$, i.e., $C^A$ is a graph invariant.
\end{theorem}
\begin{proof}
$C^A$ is sample covariance matrix of $M \in \mathbb{R}^{n \times k}$ and hence $C^A$ is symmetric positive semidefinite. Using the identity$ P^T = P^{-1}$, it is not difficult to show that for any permutation matrix P, $(PAP^T)^k = PA^kP^T$. This along with the fact $P^T\times e = e$, yields \begin{equation}\label{eq:setVec}(PAP^T)^ie = P \times A^ie.\end{equation} Thus, $C^{PAP^T}_{i,j} = Cov(P \times A^ie,P\times A^je)$. The proof follows from the fact that shuffling vectors under same permutation does not change the value of covariance between them, i.e., $$Cov(x,y) = Cov(P \times x,P \times y)$$ which implies $C^A _{i,j}= C^{PAP^T}_{i,j} \forall i,j$
\end{proof}

\noindent Note that the converse of Theorem~\ref{the:inv} is not true. We can not hope for it  because then we would have solved the intractable \emph{Graph Isomorphism Problem} by using this tractable matrix representation. For example, consider adjacency matrix of a regular graph. It has $e$ as one of its eigenvectors with eigenvalue equal to $d$, the constant degree of the regular graph. So, we have $A^ie = d^ie$ and $Cov(d^ie,d^je) = 0$. Thus, all regular graphs are mapped to the same zero matrix. Perfectly regular graphs never occur in practice, there is always some variation in the degree distribution of real-world graphs. For non regular graphs, i.e., when $e$ is not a eigenvector of the adjacency matrix, we will show in the Section~\ref{sec:proper} that the proposed $C^A$ representation is  informative.

\noindent\textbf{Alternate Motivation: Graphs as a Set of Vectors}.
There is an alternate way to motivate this representation and Theorem~\ref{the:inv}. At time $t =0$, we start with a value of $1$ on each of the nodes. At every time step $t$ we update every value on each node to the sum of numbers, from time $t-1$, on each of its neighbors. It is not difficult to show that under this process, for Node $i$, at time step $t$ we obtain $A^te(i)$. These kind of updates are key in many link analysis algorithms including Hyper-text Induced Topic Search (HITS)~\cite{kleinberg1999authoritative}. Ignoring normalization the sequence of numbers obtained over time, by such process, on node $i$ corresponds to the row $i$ of the matrix $M$. Eq. (\ref{eq:setVec}) simply tells us that reordering of nodes under any permutation does not affect the sequence of these numbers generated on each node.

Hence, we can associate a set of $n$ vectors, the $n$ rows of  $M \in\mathbb{R}^{n \times k}$, with graph $G$. This set of vectors do not change with reordering of nodes, they just shuffle among themselves. We are therefore looking for a mathematical representation that describes this set of $n$ ($k$ dimensional) vectors. Probability distributions, in particular Gaussian, are a natural way to model a set of vectors~\cite{KondorJ03}. The idea is to find the maximum likelihood Gaussian distribution fitting the given set of vectors and use this distribution, a mathematical object, as the required representation. Note that this  distribution is invariant under the ordering of vectors, and hence we get Theorem~\ref{the:inv}. The central  component of a multivariate Gaussian distribution is its covariance matrix and this naturally motivate us to study the object $C^A$, which is the covariance matrix of row vectors in $M$  associated with the graph.

\section{More Properties of Matrix $C^A$}
\label{sec:proper}
In this section, we argue that $C^A$ encodes  key features of  the given graph, making it an informative representation. In particular, we show that $C^A$ contains information about the spectral properties of $A$ as well as the counts of small substructures present in the graph. We assume that the graph is not perfectly regular, i.e., $e$ is not one of the eigenvectors of $A$. This is a  reasonable assumption because in real networks there are always  fluctuations in the degree distribution.

\vspace{0.08in}

We first start by showing connections between the matrix $C^A$  and the spectral properties of $A$. See Section~\ref{sec_notation} for the notation, for example, $\lambda_t$ and $s_t$.

\begin{theorem}
\label{theo:cij}
$C^A_{i,j} = \left(\frac{n\left(\sum_{t = 1}^n \lambda_t^{i+j}s_t^2\right)}{\left(\sum_{t = 1}^n \lambda_t^{i}s_t^2\right)\left(\sum_{t = 1}^n \lambda_t^{j}s_t^2\right)}\right) -1$
\end{theorem}
\begin{proof}
The mean of vector $A^ie$ can be written as $\frac{[e^TA^ie]}{n}$. With this observation the covariance between normalized $A^ie$ and $A^je$ (which is  equal to  $C^A(i,j)$) can be written as
\begin{align*}
Cov(A^ie,A^je) &=  \frac{1}{n}\left(n\frac{A^ie}{[e^TA^ie]} - e\right)^T\left(n\frac{A^je}{[e^TA^je]} - e\right)\\
&= \frac{1}{n}\left(n^2\frac{[e^TA^{i+j}e]}{[e^TA^ie][e^TA^je]} - n - n +e^Te\right)\\
&=  \left(n\frac{[e^TA^{i+j}e]}{[e^TA^ie][e^TA^je]}\right) - 1
\end{align*}
Thus, we have
\begin{equation}
\label{eq:cij}
C^A_{i,j} = \left(n\frac{[e^TA^{i+j}e]}{[e^TA^ie][e^TA^je]}\right) - 1
\end{equation}
To compute $[e^TA^ie]$, we use the fact that the vector $A^ie$ can be written in terms of eigenvalues and eigenvectors of $A$ as \begin{equation}\label{eq:2}A^ie = [s_1\lambda_1^i]v_1 + [s_2\lambda_2^i]v_2 + ... + [s_n\lambda_n^i]v_n.\end{equation}
This follows from the representation of $e$ in the eigenbasis of $A$, i.e., $e = s_1v_1 + s_2v_2 + ... + s_nv_n$. Using the eigenvector property $A^iv_t = \lambda^tv_t$, we have
\begin{align*}
[e^TA^ie] &= \sum_{t = 1}^n \lambda_t^i s_t [e^Tv_t]
= \sum_{t =1}^n \lambda_t^i s_t^2
\end{align*}

Substituting this value for terms $[e^TA^ie]$ in Eq. (\ref{eq:cij}) leads to the desired expression.
\end{proof}

\noindent{\bf Remarks on Theorem~\ref{theo:cij}:} We can see that different elements of matrix $C^A$  are ratios of polynomial expressions in $\lambda_t$ and $s_t$.
Given $C^A$, recovering values of $\lambda_t$ and $s_t$ $\forall\ t$ boils down to solving a set of nonlinear polynomial equations of the form given in Theorem~\ref{theo:cij} for different values of $i$ and $j$.  For a given value of $k$, we obtain a set of $\frac{k(k+1)}{2}$ different such equations. Although it may be hard to characterize the solution of this set of equations, but we can not expect many combinations of $\lambda_t$ and $s_t$ to satisfy all such equations, for some reasonably large value of $\frac{k(k+1)}{2}$.  Thus $C^A$ can be thought of as an almost lossless encoding of $\lambda_t$ and $s_t$ $\forall\ t$.

It is known that there is sharp concentration of eigenvalues of adjacency matrix  $A$ for random graphs~\cite{chung2003spectra}. The eigenvalues of adjacency matrix for a random Erdos-Reyni graph follows Wigner's semi-circle law~\cite{wigner1958distribution} while for power law graphs these eigenvalues obeys power law~\cite{chung2003spectra}.  These peculiar distributions of the eigenvalues are captured in the elements of $C^A_{i,j}$ which are the ratios of different polynomials in $\lambda_i$. Hence we can expect the $C^A$ representations, for graphs having different spectrum, to be very different.

In Theorem~\ref{theo:cij}, we have shown that the representation $C^A$ is tightly linked with the spectrum of adjacency matrix $A$, which is an important characteristic of the given graph. It is further known that the counts of various small local substructures contained in the graph such as the number of triangles, number of small paths, etc., are also important features~\cite{ugander2013}.  We next show that the matrix $C^A$ is actually sensitive to these counts.

\begin{theorem}
\label{theo:traingle}
Given the adjacency matrix $A$ of an undirected graph with $n$ nodes and $m$ edges, we have \\
$$C^A_{1,2} = \frac{n}{2m}\left(\frac{3\Delta + P_3 + n(Var(deg))+ m\left(\frac{4m}{n} -1\right)}{(P_2 +m)}\right) -1$$\\
where $\Delta$ denotes the total number of triangles, $P_3$ is the total number of distinct \emph{simple paths} of length 3, $P_2$ is the total number of distinct \emph{simple paths} of length 2  and
$$Var(deg) = \frac{1}{n} \sum_{i=1}^n deg(i)^2  - \left(\frac{1}{n}\sum_{i =1}^n deg(i)\right)^2$$ is the variance of degree.
\end{theorem}
\begin{proof}

From Eq. (\ref{eq:cij}), we have  \begin{equation}\label{eq:c12}C^A_{1,2} = \left(n\frac{[e^TA^3e]}{[e^TAe][e^TA^2e]}\right) - 1\end{equation}
The term $[e^TAe]$ is the sum of all elements of adjacency matrix $A$, which is equal to twice the number of edges. So,
\begin{equation}[e^TAe] = 2m,\end{equation}
We need to quantify other terms $[e^TA^2e]$ and $[e^TA^3e]$. This quantification is provided in the two Lemmas below.
\begin{lemma}
\label{lem:pat2}
$$[e^TA^2e] = 2m + 2P_2.$$
\end{lemma}
\begin{proof}
We start with a simple observation that the value of $A^2_{i,j}$ is equal to the number of paths of length 2 between $i$ and $j$. Thus, $[e^tA^2e]$, which is the sum of all the elements of $A^2$, counts all possible paths of length 2 in the (undirected) graph twice. We should also have to count paths of length 2 with repeated nodes because undirected edges go both ways. There are two possible types of paths of length 2 as shown in Figure~\ref{fig:path2}: i) Node repeated paths of length 2 and ii) \emph{simple paths} of length 2 having no node repetitions.

Node repeated paths of length 2 have only one possibility. It  must be a loop of length 2, which is just an edge as shown in Figure~\ref{fig:path2}(a). The total contribution of such node repeated paths (or edges) to $[e^TA^2e]$ is 2m.  By our notation, the total number of \emph{simple paths} of length 2  (Figure~\ref{fig:path2}(b)) in the given graph is $P_2$.  Both  sets of paths are disjoint. Thus, we have  $[e^TA^2e] = 2m + 2P_2$ as required.
\end{proof}

\begin{figure}[t!]
\vspace{-0.6in}
\begin{center}
\mbox{
\includegraphics[width=3in]{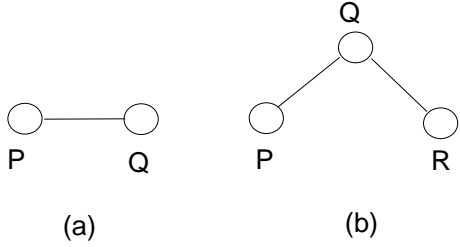}}
\end{center}
\vspace{-0.7in}
\caption{Possible types of paths of length 2, each of these two structures is counted twice in the expression $[e^TA^2e]$. a) (Node Repeated Paths): Every edge leads to two paths $P \rightarrow Q \rightarrow P$ and $Q \rightarrow P \rightarrow Q$   b) ({Simple paths}): Every \emph{simple path} of length 2 is counted twice, here $P \rightarrow Q \rightarrow R$ and $R \rightarrow Q \rightarrow P$ are the two paths contributing to the term $[e^TA^2e]$.}\label{fig:path2}
\end{figure}

\begin{lemma}
\label{lem:pat3}
\begin{align}\notag
&[e^TA^3e] = 6\Delta + 2P_3 + 2n(Var(deg))+ 2m\left(\frac{4m}{n} -1\right),\\\notag
&\text{where } \hspace{0.2in} Var(deg) = \frac{1}{n} \sum_{i=1}^n deg(i)^2  - \left(\frac{1}{n}\sum_{i =1}^n deg(i)\right)^2
\end{align}
\end{lemma}

\begin{proof}
On similar lines as Lemma~\ref{lem:pat2},  $A^3_{i,j}$ counts number of different paths of length 3. There are 3 different kinds of paths of length 3, as explained in Figure~\ref{fig:path3}, which we need to consider. We can count the contribution from each of these types independently as their contributions do no overlap and so there is no double counting. Again $[e^TA^3e]$ is twice the sum of the total number of all such paths.

{\bf Simple paths:} Just like in Lemma~\ref{lem:pat2}, any \emph{simple path} without node repetition (Figure~\ref{fig:path3}(c))  will be counted twice  in the term  $[e^TA^3e]$. Their total contribution to $[e^TA^3e]$ is $2P_3$. $P_3$ is the total number of \emph{simple paths} with length 3.

{\bf Triangles:} A triangle is the only possible loop of length 3  in the graph and it is counted 6 times in the term  $[e^TA^3e]$. There are two orientations in which a triangle can be counted from each of the three participating nodes, causing a factor of 6.  For instance in Figure~\ref{fig:path3}(b), from node $P$ there are 2 loops of length 3 to itself, $P\rightarrow R\rightarrow Q\rightarrow P$ and $P\rightarrow Q\rightarrow R\rightarrow P$. There are 2 such loop for each of the contributing nodes $Q$ and $R$. Thus, if $\Delta$ denotes the number of different triangles in the graph, then this type of structure will contribute $6\Delta$ to the term $[e^TA^3e]$.

{\bf Node Repeated Paths:} A  peculiar set of paths of length 3 are generated because of an edge $(i,j)$. In  Figure~\ref{fig:path3}(a), consider nodes $P$ and $Q$, there are many paths of length 3 with repeated nodes between $P$ and $Q$.  To go from $P$ to $Q$, we can choose any of the neighbors of $Q$, say $V$ and then there is a corresponding path $P\rightarrow Q \rightarrow V \rightarrow Q$. We can also choose any neighbor of $P$,  say $R$ and we have a path $P\rightarrow R \rightarrow P \rightarrow Q$  of length 3. Thus, given an edge $(i,j)$, the total number of  node repeated paths of length 3 is
$$NodeRepeatedPath(i,j) = deg(i) + deg(j) -1.$$
Note that the path $P \rightarrow Q \rightarrow P \rightarrow Q$, will be counted twice and therefore we subtract 1. Thus, the total contribution of these kinds of paths in the term $[e^TA^3e]$ is
$$\sum_{(i,j) \in E} (deg(i) + deg(j) -1),$$
Since the graph is undirected both $(i,j) \in E \implies (j,i) \in E$, so we do not have to use a factor of 2 like we did in other cases.
We have
\begin{align}\notag
&\sum_{(i,j) \in E} (deg(i) + deg(j) -1) = \sum_{i =1}^n \sum_{j \in Ngh(i)} (deg(i) + deg(j) -1)\\\notag
&= \sum_{i=1}^n (deg(i)^2 - deg(i)) + \sum_{i=1}^n  \sum_{j \in Ngh(i)} deg(j) \\\notag
&= 2\sum_{i=1}^n deg(i)^2 - \sum_{i=1}^ndeg(i)
\end{align}
Adding contributions of all  possible types of paths and using $\sum_{i=1}^ndeg(i) = 2m$ yields Lemma~\ref{lem:pat3} after some algebra.
\end{proof}

\begin{figure}[t!]
\vspace{-0.3in}
\begin{center}
\mbox{
\includegraphics[width=3in]{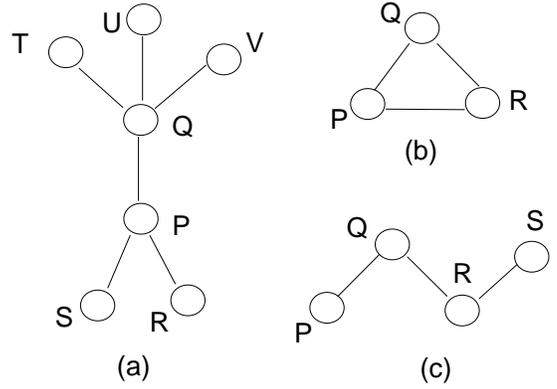}}
\end{center}
\vspace{-0.4in}
\caption{All possible types of paths of length 3, from different  structures contributing to the term $[e^TA^3e]$. a) (Node Repeated Paths): There are 6 paths of length 3 from $P$ to $Q$ (and vice versa) with repeated nodes like $P\rightarrow Q \rightarrow V \rightarrow Q$. The total number of paths of length 3 due to edge between $i$ and $j$ is equal to $deg(i) + deg(j) - 1$.  b) (Triangles): A triangle is counted 6 times in the expression $[e^TA^3e]$, in two different orientations from each of three nodes. c) (Simple Paths):  A simple paths with no node repetition of length 3 will be counted twice in $[e^TA^3e]$. }\label{fig:path3}
\end{figure}

Substituting for the terms $[e^TA^2e]$ and $[e^TA^3e]$ in Eq. (\ref{eq:c12}) from Lemmas~\ref{lem:pat2} and~\ref{lem:pat3}  leads to the desired expression.
\end{proof}

\noindent{\bf Remarks on Theorem~\ref{theo:traingle}:} From its proof, it is clear that terms of the form $[e^TA^te]$, for small values of $t$ like 2 or 3,  are weighted combinations of counts of small sub-structures like triangles and small paths along with global features like degree variance. The key observation behind the proof is that $A^t_{i,j}$ counts paths (with repeated nodes and edges) of length $t$, which in turn can be decomposed into disjoint structures over $t+1$ nodes and can be counted separately.  Extending this analysis for $t > 3$, involves dealing with more complicated bigger patterns. For instance, while computing the term  $[e^TA^4e]$, we will encounter counts of quadrilaterals along with more complex patterns.  The representation $C^A$ is informative in that it captures all such information and is sensitive to the counts of these different  substructures present in the graph.

\noindent{\bf Empirical Evidence for Theorem~\ref{theo:traingle}:} To empirically validate Theorem~\ref{theo:traingle}, we took publicly available twitter graphs\footnote{http://snap.stanford.edu/data/egonets-Twitter.html}, which consist of around 950 ego networks  of users on twitter~\cite{mcauley2012}. These graphs have around 130 nodes and 1700 edges on an average. We computed the value of $\Sigma_e^A(1,2)$ for each  graph (and the mean and  standard error). In addition, for each twitter graph, we also generated a corresponding random graph with same number of nodes and edges. To generate a random graph, we start with the required number of nodes and then select two nodes at random and add an edge between them. The process is repeated until the graph has the same number of edges as the  twitter graph. We then compute the value of $\Sigma_e^A(1,2)$ for all these generated random graphs. The mean ($\pm$ standard error, SE) value of $\Sigma_e^A(1,2)$ for twitter graphs is {0.6188 $\pm$ 0.0099}, while for the random graphs this value  is {0.0640 $\pm$ 0.0033}.

The mean ($\pm$ SE) number of triangle for twitter ego network is {14384.16 $\pm$ 819.39}, while that for random graphs is {4578.89 $\pm$ 406.54}.  It is known that  social network graphs have a high value of  {\em triadic closure probability} compared to random graphs~\cite{easley2012networks}. For any 3 randomly chosen vertices A, B and C in the graph, {triadic closure probability} (common friendships induce new friendships) is a probability of having an edge AC conditional on the event that the graph already has edges AB and BC. Social network graphs have more triangles compared to a random graph. Thus,  Theorem~\ref{theo:traingle} suggests that the value of $\Sigma_e^A(1,2)$ would be high for a social network graph compared to a random graph with same number of nodes and edges.

Combining Theorems \ref{theo:cij} and \ref{theo:traingle}, we can infer that our proposed representation $C^A$ encodes  important information to discriminate between different network structures.  Theorem~\ref{the:inv} tells us that this object is a graph invariant and a covariance matrix in a fixed dimensional space. Hence $C^A$ is directly comparable between different graph structures.

\section{Similarity between Graphs}

Given a fixed $k$, we have a representation for graphs in a common mathematical space, the space of symmetric positive semidefinite matrices $\mathbb{S}_{k \times k}$, whose mathematical properties are well understood. In particular, there are standard notions of similarity between such matrices. We define similarity  between two graphs, with adjacency matrices $A \in \mathbb{R}^{n_1 \times n_1}$ and $B \in \mathbb{R}^{n_2 \times n_2}$ respectively,  as the Bhattacharya similarity between corresponding covariance matrices $C^A$ and $C^B$ respectively:
\begin{align}
\label{eq:kernel}
Sim(C^A,C^B) &= exp^{-Dist(C^A,C^B)}\\
Dist(C^A,C^B) &=  \frac{1}{2}\log\left(\frac{det(\Sigma)}{\sqrt{(det(C^A)det(C^B))}}\right)\notag\\
\Sigma &= \frac{C^A + C^B}{2}\notag
\end{align}
Here, $det()$ is the determinant. Note that $C^A \in \mathbb{R}^{k \times k}$ and $C^B  \in \mathbb{R}^{k \times k}$ are computed using the same value of $k$.  We summarize the procedure of computing similarity between two graphs with adjacency matrices $A$ and $B$ in Algorithm~\ref{alg:sim}.

\begin{algorithm}[h!]
\caption{\emph{ComputeSimilarity(A,B,k)}}
\label{alg:sim}
\begin{algorithmic}
\vspace{0.05in}
\STATE {\bfseries Input:} Adjacency matrices $A\in\mathbb{R}^{n_1 \times n_1}$ and $B \in \mathbb{R}^{n_2 \times n_2}$, $k$, the number of power iterations.
\vspace{0.05in}
\STATE $C^A =  CovarianceRepresentation(A,k)$
\vspace{0.05in}
\STATE $C^B = CovarianceRepresentation(B,k)$
\vspace{0.05in}
\RETURN  $Sim(C^A,C^B)$ computed using Eq. (\ref{eq:kernel})
\end{algorithmic}
\end{algorithm}

\begin{theorem}
The similarity $Sim(C^A,C^B)$, defined between graphs with adjacency matrices $A$ and $B$, is positive semidefinite and is a valid kernel.
\end{theorem}
This similarity is positive semidefinite, which follows  from the fact that the Bhattacharya similarity is positive semidefinite. Thus, the similarity function defined in Eq. (\ref{eq:kernel}) is a valid kernel~\cite{hofmann2008kernel} and hence can be directly used in existing machine learning algorithms operating over kernels such as SVM. We will see performance of this kernel on the task of social network classification later in Section~\ref{sec:exp}.

Although  $C^A$ is determined by the spectrum of adjacency matrix $A$, we will see  in Section~\ref{sec:result}, that simply taking a feature vector of graph invariants such as eigenvalues and computing the vector inner products is not the right way to compute similarity between graphs.  It is crucial to consider the fact that we are working in the space of positive semidefinite covariance matrices and  a similarity measure should utilize the mathematical structure of the space under consideration.

\begin{table*}[t!]
\centering
\caption{Graph statistics of ego-networks used in the paper. The ``Random'' datasets consist of random Erdos-Reyni graphs (see Section~\ref{sub:data} for more details)}
\begin{tabular}{|c|c|c|c|c|l|} \hline
STATS  & High Energy & Condensed Matter & Astro Physics & Twitter & Random\\\hline
Number of Graphs & 1000 & 415&1000 & 973 & 973\\\hline
Mean Number of Nodes &131.95 & 73.87 & 87.40 &137.57 &137.57 \\\hline
Mean Number of Edges & 8644.53 &410.20 & 1305.00 &1709.20 & 1709.20\\\hline
Mean Clustering Coefficient & 0.95& 0.86& 0.85& 0.55& 0.18 \\
\hline\end{tabular}
\label{tab:stat}
\end{table*}

\subsection{Range for Values of $k$}
\label{sec:k}
Our representation space, the space of symmetric positive semidefinite matrices, $\mathbb{S}_{k \times k}$ is dependent on the choice of $k$. In general, we only need to look at small values of $k$. It is known that power iteration converges at a geometric rate of $\frac{\lambda_2}{\lambda_1}$ to the largest eigenvector of the matrix, and hence covariance between normalized $A^ie$ and $A^je$ will converge to a constant very quickly as the values of $i$ and $j$ increase. Thus, large values of $k$ will make the matrix singular and hurt the representation.  We therefore want the value of $k$ to be reasonably small to avoid singularity of matrix $C^A$. The exact choice of $k$ depends on the dataset under consideration. We observe $k=4\sim6$  suffices in general.

\subsection{Computation Complexity}
\label{sec:CComp}
For a chosen $k$, computing the set of vectors $\{Ae, A^2e,$  $A^3e, ..., A^ke\}$ recursively as done in Algorithm~\ref{alg:cov} has computation complexity of $O(mk)$. Note that the number of nonzeros in matrix $A$ is $2m$ and each operation inside the for-loop is a sparse matrix vector multiplication, which has complexity $O(m)$.  Computing  $C^A$ requires summation of $n$ outer products of vectors of dimension $k$, which has complexity $O(nk^2)$. The total complexity of Algorithm~\ref{alg:cov} is $O(mk + nk^2)$.

Computing similarity between two graphs, with adjacency matrices $A$ and $B$ in addition requires computation of Eq.~(\ref{eq:kernel}), which involves computing determinants of $k \times k$ matrices. This operation has computational complexity $O(k^3)$. Let the number of nodes and edges in the two graphs be $(n_1, m_1)$ and $(n_2,m_2)$ respectively. Also, let $m =\max(m_1,m_2)$ and $n = \max(n_1,n_2)$. Computing similarity using Algorithm~\ref{alg:sim} requires $ O(mk + nk^2 + k^3)$ computation time.

As argued in Section~\ref{sec:k}, the value of $k$ is always a small constant like 4, 5 or 6. Thus, the total time complexity of computing the similarity between two graphs reduces to $O(m + n) = O(m)$ (as usually $m \ge n$). The most costly step is the matrix vector multiplication which can be easily parallelized, for example  on GPUs, to obtain further speedups. This makes our proposal easily scalable in practice.

\section{Social Network Classification}
\label{sec:exp}
In this section, we demonstrate the usefulness of the proposed representation for graphs and the new similarity measure in some interesting graph classification tasks. We start by describing these tasks and the corresponding datasets.

\subsection{Task and Datasets}

\label{sub:data}
Finding publicly available  datasets for graph classification task, with meaningful label, is difficult in the domain of social networks. However, due to the increasing availability of many different network structures\footnote{http://snap.stanford.edu/data/} we can create interesting and meaningful classification tasks. We create two social network classification tasks from real networks.

\noindent{\bf 1. Ego Network Classification in Scientific Collaboration (COLLAB):} Different research fields have different collaboration patterns.  For instance,  researchers in experimental high energy physics are dependent on few specialized labs worldwide (e.g., CERN).  Because of this dependency on specialized labs, various  research groups in such domains are tightly linked in terms of collaboration compared to other domains where more independent research is possible. It is an interesting  task to classify the research area of an individual by taking into account the information contained in the structure of  his/her ego collaboration network.

We used 3 public collaboration network datasets~\cite{leskovec2007graph}: 1) High energy physics collaboration network\footnote{http://snap.stanford.edu/data/ca-HepPh.html}, 2) Condensed matter physics collaboration network\footnote{http://snap.stanford.edu/data/ca-CondMat.html}, 3) Astro physics collaboration network.\footnote{http://snap.stanford.edu/data/ca-AstroPh.html} These networks are generated from e-print arXiv and cover scientific collaborations between author's papers submitted to  respective categories. If author $i$ co-authored a paper with author $j$, the graph contains an undirected edge from $i$ to $j$. If the paper is co-authored by $p$ authors, this generates a completely connected subgraph on $p$ nodes.

To generate meaningful ego-networks from each of these huge collaboration networks, we select different users who have collaborated with more than 50 researchers and extract their ego networks. The ego network is the subgraph containing the selected node along with its neighbors and all the interconnections among them. We randomly choose 1000 such users from each of the high energy physics collaboration network and the astro physics collaboration network. In case of condensed matter physics, the collaboration network only had 415 individuals with more than 50 neighbors we take all the available 415 ego networks.

In this way, we obtain 2415 undirected ego network structures. The basic statistics of these ego networks are summarized in Table~\ref{tab:stat}. We label each of the graphs according to which of the three collaboration network it belongs to.  Thus, our classification task is to take a researcher's ego collaboration network and determine whether he/she belongs to high energy physics group, condensed matter physics group, or Astro physics group. This is a specific version of a more general problem that arises in social media:  ``how audiences differ with respect to their social graph structure ?" ~\cite{adamic2008knowledge}.

For better insight into performance, we break the problem class-wise into 4 different classification tasks: 1) classifying between high energy physics and condensed matter physics (COLLAB (HEP Vs CM)) 2) classifying between  high energy physics and  astrophysics (COLLAB (HEP Vs ASTRO)) 3) classifying between  astrophysics and  condensed matter physics (COLLAB (ASTRO Vs CM)) and 4) classifying among all the three domains (COLLAB (Full)).

\noindent{\bf 2. Social Network Classification (SOCIAL):} It is known that social network graphs behave very differently from random Erdos-Reyni graphs~\cite{watts1998collective}. In particular,   a random Erdos-Reyni graph does not have the following two important properties observed in many real-world networks:
\begin{itemize}
\item  They do not generate local clustering and triadic closures. Because they have a constant, random, and independent probability of two nodes being connected, Erdos-Reyni graphs have a low clustering coefficient.
\item  They do not account for the formation of hubs. Formally, the degree distribution of Erdos-Reyni random graphs converges to a Poisson distribution, rather than a power law observed in many real-world networks.
\end{itemize}

Thus, one reasonable  task is to discriminate social network structures from random Erdos-Reyni graphs. We expect methodologies which capture properties like triadic closure and the degree distribution to perform well on this task.

We used the Twitter\footnote{http://snap.stanford.edu/data/egonets-Twitter.html} ego networks~\cite{mcauley2012}, which is a large  public  dataset of social ego networks, which contains around 950 ego networks of users from Twitter with a mean of around 130 nodes and 1700 edges per graph.  Since we are interested only in the graph structure, these directed graphs were made undirected. We do not use any information other than the adjacency matrix of the graph structure.

For each of the undirected twitter graphs, we generated a corresponding random graph with the same number of nodes and edges. We start with the required number of nodes and then select two nodes at random and add an edge between them. This process is repeated until the graph has the same number of edges as the corresponding twitter graph.  We label these graphs according to whether they are a Twitter graph or a random graph. Thus, we have a binary classification task consisting of around 2000 graph structures. The basic statistics of this dataset are summarized in Table~\ref{tab:stat}.

\subsection{Competing Methodologies}
\label{sec:method}
For classification task it suffices to have a similarity measure (commonly known as kernel) between two graphs, which is positive semidefinite. Our evaluation consists of running standard kernel $C$-SVMs~\cite{CC01a} for classification, based on the following five similarity measures.

\begin{table*}
\centering
\caption{Prediction accuracies in percentage for proposed and the state-of-the-art similarity measures
on different social network classification tasks. The reported results are averaged over 10 repetitions of
10-fold cross-validation. Standard errors are indicated using parentheses. Best results marked in \textbf{bold}.}
\begin{tabular}{|p{2.9cm}|p{2.5cm}|p{2.5cm}|p{2.5cm}|p{2.5cm}|p{2.5cm}|} \hline
Methodology  & COLLAB\hspace{0.5in} (HEnP Vs CM) & COLLAB\hspace{0.5in} (HEnP Vs ASTRO) & COLLAB\hspace{0.5in} (ASTRO Vs CM) & COLLAB (Full) & SOCIAL\hspace{0.5in} (Twitter Vs Random) \\\hline
PROPOSED (k =4) & 98.06(0.05) & {\bf 87.70(0.13)}& 89.29(0.18)& 82.94(0.16) & 99.18(0.03)\\\hline
PROPOSED (k =5) & {\bf 98.22(0.06)} &87.47(0.04) &89.26(0.17) & {\bf 83.56(0.12)} & 99.43(0.02)\\\hline
PROPOSED (k =6) & 97.51(0.04) & 82.07(0.06)& {\bf 89.65(0.09)}& 82.87(0.11) & {\bf 99.48(0.03)} \\\hline
SUBFREQ-5 & 96.97 (0.04)& 85.61(0.1)& 88.04(0.14)& 81.50(0.08)& 99.42(0.03) \\\hline
SUBFREQ-4 & 97.16 (0.05)& 82.78(0.06)& 86.93(0.12)& 78.55(0.08)& 98.30(0.08) \\\hline
SUBFREQ-3 & 96.38 (0.03)&80.35(0.06) & 82.98(0.12) & 73.42(0.13)& 89.70(0.04)\\\hline
RW & 96.12 (0.07)& 80.43(0.14)& 85.68(0.03)& 75.64(0.09)& 90.23(0.06)\\\hline
EIGS-5 & 94.85(0.18)& 77.69(0.24)& 83.16(0.47)& 72.02(0.25)& 90.74(0.22) \\\hline
EIGS-10 & 96.92(0.21)& 78.15(0.17)& 84.60(0.27)& 72.93(0.19)&92.71(0.15) \\
\hline\end{tabular}
\label{tab:result}
\end{table*}

\noindent{\bf The Proposed Similarity (PROPOSED):} This is the proposed similarity measure. For the given two graphs, we compute the similarity between them using Algorithm~\ref{alg:sim}. We show results for 3 fixed values of $k = \{4,5,6\}$.

\noindent{\bf 4-Subgraph Frequency (SUBFREQ-4):} Following~\cite{ugander2013}, for each of the graphs we first generate a feature vector of normalized frequency of subgraphs of size four. It is  known that the subgraph frequencies of arbitrarily large graphs can be accurately approximated by sampling a small number of induced subgraphs. In line with the recent work, we computed such a histogram by sampling 1000 random subgraphs over 4 nodes. We observe that 1000 is a  stable sample size and  increasing this number  has almost no effect on the accuracy.  This process generates a normalized  histograms of dimension 11 for each graph as there are 11 non-isomorphic different graphs with 4 nodes (see~\cite{ugander2013} for more details). The similarity value between two graphs is the inner product between the corresponding  11 dimensional vectors.

\noindent{\bf 5-Subgraph Frequency (SUBFREQ-5):} Recent success of counting induced subgraphs of size 4 in the domain of social networks leads to a  natural curiosity ``whether counting all subgraphs of size 5 improves the accuracy values over only counting subgraphs of size 4?" To answer this, we also consider the histogram of normalized frequency of subgraphs of size 5. Similar to the case of SUBFREQ-4, we sample 1000 random induced subgraphs of size 5 to generate a histogram representation. There are 34 non-isomorphic different graphs on 5 nodes and so this procedure generates a vector of 34 dimensions  and the similarity  between two graphs is the inner product between the corresponding 34-dim feature vectors. Even with sampling, this is an expensive task and takes significantly more time than  SUBFREQ-4.  The main reason for this is the increase in the number of isomorphic variants. Matching a given sampled graph to one of the representative structure is actually solving graph isomorphism over graphs of size 5, which is costly (see Section~\ref{sec:comp}).

\noindent{\bf 3-Subgraph Frequency (SUBFREQ-3):} To quantify the importance of size 4 subgraphs, we also compare with the histogram representation based on frequencies of subgraphs of size 3. There are  4 non-isomorphic different graphs with 3 nodes and hence here we generate a histogram of dimension 4. As counting subgraphs of size 3 is computationally cheap we do not need sampling for this case. This simple representation is known to perform quite well in practice~\cite{shervashidze2009efficient}.

\noindent{\bf Random Walk Similarity (RW):}  Random walk similarity is one of the widely used similarity measures over graphs~\cite{DasSarma:2013,vishwanathan2010graph}. It is based on a simple idea: given a pair of graphs,
perform random walks on both, and count the number of similar walks. There is a rich set of literature regarding connections of this similarity with well-known similarity measures in different domains such as Binet-Cauchy Kernels for ARMA models~\cite{vishwanathan2004binet}, rational kernels~\cite{cortes2002rational}, r-convolution kernels~\cite{haussler1999convolution}. The random walk similarity~\cite{vishwanathan2004binet} between two graphs with adjacency matrix $A$ and $B$ is defined as
$$RWSim(A,B) =
\frac{1}{n_1n_2}e^TMe,$$ where $M$ is the solution of Sylvester equation $$M = (A^TMB)exp^{-\nu}+ ee^T.$$ This can be   computed in closed-form in $O(n^3)$ time. We use standard recommendations for choosing the value of $\nu$.

\noindent{\bf Top-$k$ Eigenvalues (EIGS):} It is known that the eigenvalues of the adjacency matrix are the most important graph invariants. Therefore it is worth considering the power of simply using the dominant eigenvalues.  Note that we can not take all eigenvalues because the total number of eigenvalues varies with the graph size. Instead, we take top-$k$ eigenvalues of the corresponding adjacency matrices and compute the normalized inner product between them.  We show the results for $k = 5$ (EIGS-5) and $k=10$ (EIGS-10).

\subsection{Evaluations and Results}
\label{sec:result}
The evaluations consist of running kernel SVM on all the tasks using  six different   similarity measures as described above,   based on the standard cross-validation estimation of classification accuracy. First, we split each dataset into 10 folds of identical size. We then combine 9 of these folds and again split it into 10 parts, then use the first 9 parts to train the kernel $C$-SVM~\cite{CC01a} and use the 10th part as validation set to find the best performing value of C from $\{10^{-7},10^{-6},...,10^7\}$. With this  fixed choice of $C$, we then train the $C$-SVM on all the 9 folds (from initial 10 folds) and predict on the 10th fold acting as an independent evaluation set. The procedure is repeated 10 times with each fold acting as an independent test set once. For each task, the  procedure is  repeated 10 times randomizing over partitions. The mean classification accuracies and the standard errors are shown in Table~\ref{tab:result}. Since we have not tuned anything other than the  ``$C$" for  SVM,   the results are easily reproducible.

In those tasks, using our proposed representation and similarity measure outperforms all the  competing state-of-the-art methods, mostly with a significant margin.  This  demonstrates that the covariance matrix representation captures sufficient information about the ego networks and is capable of discriminating between them. The accuracies for three different values of $k$ are not  much different form each other, except in some cases with $k =6$. This is in line with the  argument presented in Section~\ref{sec:k} that large values of $k$ can hurt the representation. As long as $k$ is small and is in the right range, slight variations in $k$ do not have  significant change in the performance. Ideally, $k$ can be tuned based on the dataset, but for easy replication of results we used 3 fixed choices of $k$.

We see that random walk similarity performs similarly (sometimes better) to  SUBFREQ-3  which counts all the subgraphs of size 3. The performance of EIGS is very much like the random walk similarity. As expected (e.g., from the recent work~\cite{ugander2013}), counting subgraphs of size 4 (SUBFREQ 4) always improve significantly over SUBGREQ-3. Interestingly, counting subgraphs of size 5 (SUBFREQ-5) improves significantly over SUBFREQ-4 on all tasks, except for HEnP Vs CM task. This illustrates the sub-optimality of histogram obtained by counting very small graphs ($k \le 4$). Even with sampling, SUBFREQ-5 is an order of magnitude more expensive than other methodologies. As shown  in the next section, with increasing $k$, we loose the computational tractability  of counting induced $k$-subgraphs (even with sampling).

Our covariance methodology consistently performs better than (SUBFREQ 5), demonstrating the superiority of the $C^A$ representation.  As argued in Section~\ref{sec:proper}, the matrix $C^A$  even for $k = 4 \mbox{ or } 5$,  does incorporate information regarding the counts of bigger complex sub-structures in the graph. This along with the information of the full spectrum of the adjacency matrix leads to a  sound representation which outperforms state-of-the-art similarity measures over graphs.

\subsection{Why Simply Computing Graph Invariants is Not Enough?}

It can be seen that vector representation of dominant eigenvalues performs very poorly compared to the proposed representation even though Theorem~\ref{theo:cij} says that every element of the proposed matrix representation is a function of eigenvalues. It is  not very clear how to compare eigenvalues across graphs. For instance, two graphs with different sizes will usually have different number of eigenvalues. A vector consisting of few dominant eigenvalues does not seem to be the right object describing graphs, although, most of the characteristics about a given graph can be inferred from it. A good analogy to explain this would be that the mean $\mu$ and variance $\sigma^2$ fully determines a Gaussian random variable, but to compute distance between two Gaussian distributions, simply computing the euclidian distance between corresponding $(\mu,\sigma^2)$ does not work well. The proposed $C^A$ representation, a graph invariant, seems a better object which being a covariance matrix is comparable and standard similarity measures over $C^A$ performs quite well. The informativeness of features is necessary but not sufficient for learning, a classical problem in machine learning where finding the right representation is the key.

\section{Running Time Comparisons}
\label{sec:comp}
To obtain an estimate of the computational requirements, we compare the time required to compute the similarity values between two given graphs using different methodologies presented in Section~\ref{sec:method}. For both  datasets, we record the cpu-time taken for computing pairwise similarity  between all possible pairs of graphs. Thus, for COLLAB dataset, this is the time taken to compute similarity between 2415(2415 -1)/2 pairs of networks while in case of SOCIAL it is the time taken to compute similarity between  1946(1946 -1)/2 pairs of networks.  The times taken by different methodologies on these two datasets are summarized in Table~\ref{tab:runtime}.  All experiments were performed in MATLAB on an Intel(R) Xenon 3.2 Ghz CPU machine having 72 GB of RAM.

EIGS, although it performs poorly in terms of accuracy, is the fastest compared to all other algorithms, because there are very fast linear algebraic methods for computing top-$k$ eigenvalues. We can see that, except for SUBFREQ-5 and RW, all other methods are quite competitive in terms of run-time. It is not surprising that RW kernels are slower because they are known to have cubic run-time complexity. From Section~\ref{sec:CComp}, we know that the proposed methodology is actually linear in $O(E)$. Also, there are very efficient ways of computing SUBFREQ-3~\cite{shervashidze2009efficient} from the adjacency list representation which is being used in the comparisons. Although computing histogram based on counting all the subgraphs of size 4 is much more  costly than counting subgraphs of size 3, approximating the histogram by sampling is fairly efficient.  For example, on the COLLAB dataset, approximating SUBFREQ-4  by taking 1000 samples is even more efficient than counting all subgraphs of size 3.

\begin{table}[h!]
\centering
\caption{Time (in sec) required for computing all pairwise similarities of the two datasets.}
\begin{tabular}{|c|c|c|} \hline
&SOCIAL & COLLAB (Full) \\ \hline
Total Number of Graphs & 1946 & 2415 \\ \hline
PROPOSED (k =4) & 177.20& 260.56 \\
PROPOSED (k =5)&200.28 & 276.77\\
PROPOSED (k =6)& 207.20& 286.87\\
SUBFREQ-5 (1000 Samp) &5678.67& 7433.41\\
SUBFREQ-4 (1000 Samp) & 193.39 & 265.77\\
SUBFREQ-3 (All)  & 115.58 & 369.83\\
RW & 19669.24&25195.54\\
EIGS-5 & 36.84 & 26.03\\
EIGS-10 & 41.15 &29.46\\
\hline\end{tabular}
\label{tab:runtime}
\end{table}

However, even with sampling, SUBFREQ-5 is an order of magnitude slower. To understand this, let us review the process of computing the histogram by counting subgraphs. There are 34 graph structures over 5 nodes unique up to isomorphism. Each of these 34 structures has $5! = 120$ many isomorphic variants (one for every permutation). To compute a histogram over these 34 structures, we first sample an induced 5-subgraph from the given graph. The next step is to match this subgraph to one of the 34 structures. This requires determining which of the 34 graphs is isomorphic with the given sampled subgraph. The process is repeated 1000 times for every sample.  Thus every sampling step requires solving graph isomorphism problem. Even  SUBFREQ-4 has the same problem but there are only 11 possible subgraphs and the number of isomorphic structures for each graph is only $4! = 24$, which is still efficient. This scenario starts becoming intractable as we go beyond 5 because of the combinatorially hard graph isomorphism problem.

SUBFREQ-5, although it is computationally very expensive, improves over SUBFREQ-4. The proposed similarity based on $C^A$ is almost as cheap as SUBFREQ-4 but performs better than even SUBFREQ-5. Counting based approaches, although they capture information, quickly loose tractability once we start counting bigger substructures. Power iteration of the adjacency matrix is a nice and computationally efficient  way of capturing information about the underlying graph.

\section{Conclusions}

We embed graphs into a new mathematical space, the space of symmetric positive semidefinite matrices $\mathbb{S}_{k \times k}$. We take an altogether different approach of characterizing graphs based on the covariance matrix of the vectors obtained from the power iteration of the adjacency matrix. Our analysis indicates that the proposed matrix representation $C^A$ contains most of the important characteristic  information about the networks structure.  Since the $C^A$ representation is a covariance matrix in a fixed dimensional space, it naturally gives a measure of similarity (or distance) between different graphs. The overall procedure is simple and scalable in that it can be computed in time linear in number of edges.

Experimental evaluations demonstrate the superiority of the $C^A$ representation, over other state-of-the-art methods,  in ego network classification tasks.  Running time comparisons indicate that the proposed approach provides the right balance between the expressiveness of representation and the computational tractability.
Finding tractable and meaningful  representations of graph is  a fundamental problem,  we believe our results as shown will provide  motivation for using the new representation  in analyzing real networks.


\end{document}